\documentclass[11pt,letterpaper]{article}
\pdfoutput=1

\usepackage{epsfig}
\usepackage{graphicx}
\usepackage{amsthm}
\usepackage{amssymb}
\usepackage{amsmath}
\usepackage{amsfonts}
\usepackage{mathrsfs}
\usepackage{fullpage}

\usepackage{color}

\newtheorem{lemma}{Lemma}
\newtheorem{theorem}{Theorem}

\newtheorem{definition}{Definition}
\newtheorem{observation}{Observation}

\newtheorem{claim}{Claim}

\newcommand{\commentout}[1]{}

\newcommand{\cls}[1]{\ensuremath{\mathsf{#1}}}
\def\NP{\cls{NP}}
\renewcommand{\P}{\cls{P}}
\def\ZPP{\cls{ZPP}}
\def\FPT{\cls{FPT}}
\def\APX{\cls{APX}}
\def\Wone{\ensuremath{\cls{W}[1]}}

\newcommand{\defer}[1]{}

\begin{document}

\title{On Approximating String Selection Problems with Outliers}
\author{Christina Boucher\thanks{Department of Computer Science and Engineering, University of California, San Diego}
\and  Gad M. Landau\thanks{Department of Computer Science, University of Haifa, Mount Carmel, Haifa 31905, Israel.}~\thanks{Polytechnic Institute of NYU,  NY 11201-3840, USA.}
\and Avivit Levy\thanks{Shenkar College for Engineering and Design, Ramat-Gan, 52526, Israel}~\thanks{CRI, University of Haifa, Mount Carmel, Haifa 31905, Israel}
%\and Daniel Lokshtanov\footnotemark[1]
\and David Pritchard\thanks{CEMC, University of Waterloo, Canada}
\and Oren Weimann\footnotemark[2]}
\maketitle

\maketitle

%Abstract

\begin{abstract}
Many problems in bioinformatics are about finding strings that approximately represent a collection of given strings. We look at more general problems where some input strings can be classified as outliers. The {\em Close to Most Strings} problem is, given a set $S$ of same-length strings, and a parameter $d$, find a string $x$ that maximizes the number of ``non-outliers" within Hamming distance $d$ of $x$. We prove this problem has no PTAS unless $\ZPP=\NP$, correcting a decade-old mistake. The {\em Most Strings with Few Bad Columns} problem is to find a maximum-size subset of input strings so that the number of non-identical positions is at most $k$; we show it has no PTAS unless $\P=\NP$. We also observe {\em Closest to $k$ Strings} has no EPTAS unless $\Wone=\FPT$. In sum, outliers help model problems associated with using biological data, but we show the problem of finding an approximate solution is computationally difficult.
\end{abstract}

%\begin{abstract}
%Lanctot et al.~(SODA'99) initiated the study of {\em distinguishing string selection problems} in bioinformatics, where the goal is to find a string subject to similarity constraints with respect to a given set of input strings.  In this paper, we investigate the computational difficulty of three distinguishing string selection problems when there are a fixed or variable number of ``outliers.'' For example, we observe that {\em Closest to $k$ Strings}, a variant of the classical {\em Closest String Problem}, has no EPTAS unless $\Wone=\FPT$. One of our main results is about the {\em Close to Most Strings Problem}: given a set $S$ of $n$ length-$\ell$ strings, and parameter $d$, find a string $x$ that maximizes the number of strings that are within Hamming distance $d$ from $x$.
%We prove that Close to Most Strings has no PTAS unless $\ZPP=\NP$, which corrects a decade-old flaw in the research literature. Our other main result is about the {\em Most Strings with Few Bad Columns} problem, where we want a maximum-size subset of input strings so that the number of \emph{bad} columns --- columns where not all strings are equal --- is at most $k$. We prove that Most Strings with Few Bad Columns has no PTAS, unless $\P=\NP$. In sum, while outliers help model the problems associated with using biological data, this addition causes the task of finding a good approximate solution to become computationally intractable.
%\end{abstract}

\section{Introduction}

With the development of high-throughput next generation sequencing technologies, there has arisen large amounts of genomic data, and an increased need for novel ways to analyze this data.  This has inspired numerous formulations of biological tasks as computational problems.  In light of this observation, Lanctot et al.~\cite{lanctot} initiated the study of {\em distinguishing string selection problems}, where we seek a representative string satisfying some distance constraints from each of the input strings. We will mostly have constraints in the form of an upper bound on the Hamming distance, but lower bounds on the Hamming distance, and substring distances, have also been considered~\cite{DLLM,GGR03,lanctot}.

%For example, the problem of finding transcription factor binding sites in biological data is abstractly viewed as the motif-recognition problem~\cite{lanctot,PS00}.  Transcription factors are proteins that bind to promoter regions in the genome and have the effect of regulating the expression of one or more genes. Hence, the region where a transcription factor binds is very well-conserved, and the problem of detecting such regions can be viewed as a stringology problem.

Typically, the distance constraint must be satisfied for each of the input strings.  However, biological sequence data is subject to frequent random mutations and errors, particularly in specific segments of the data; requiring that the solution fits the entire input data is problematic for many problems in bioinformatics. %---including the one previously discussed.
  It would be preferable to find the similarity of a portion of the input strings, excluding a few bad reads that have been corrupted, rather than trying to fit the complete set of input and in doing so finding one that is distant from many or all of the strings.

What if we are given a measure of goodness (e.g., distance) the representative must satisfy, and want to choose the largest \emph{subset} of strings with such a representative? Conversely, what if we specify the subset size and seek a representative that is as good as possible?  Some results are known in this area with respect to fixed-parameter tractability~\cite{boucher_ma}. Here, we prove results about the approximability of three {\em string selection problems with outliers}. %One of our results resolves a decade old mistake in the research literature.
For any two strings $x$ and $y$ of same length, we denote the Hamming distance between them as $d(x, y)$, which is defined as the number of mismatched positions. Our main results are about three \NP\ optimization problems. %Note that the decision versions of the first two are the same problem (i.e. the Closest String Problem).

\begin{definition} {\em Close to Most Strings} (a.k.a.~{\em Max Close String} \cite{lanctot,ma00})\\
\noindent Input: $n$ strings $S = \{s_1, \ldots, s_n\}$ of length $\ell$ over an alphabet $\Sigma$, and $d \in \mathbf{Z}_+$.\\
\noindent Solution: a string $s$ of length $\ell$.\\
\noindent Objective: maximize the number of strings $s_i$ in $S$ that satisfy $d(s, s_i) \leq d$.
\end{definition}

\begin{definition} {\em Closest to $k$ Strings} \\
\noindent Input: $n$ strings $S = \{s_1, \ldots, s_n\}$ of length $\ell$  over an alphabet $\Sigma$, and $k \in \mathbf{Z}_+$.\\
\noindent Solution: a string $s$ of length $\ell$ and a subset $S^*$ of $S$ of size $k$.\\
\noindent Objective: minimize $\max \{d(s, s_i) \mid s_i \in S^*\}$.
\end{definition}
\noindent In the special case $k=n$, {\em Closest to $k$ Strings} becomes {\em Closest String} --- an \NP-hard problem~\cite{FL97} that has received significant interest in parameterized complexity and approximability \cite{ami2,AIP,GNR03,lanctot,LMW02,MS08,wang}.

We also consider a problem where the ``outliers" are considered to be positions (``columns") rather than strings (``rows"). Let $s(j)$ indicate the $j$th character of string $s$.

\begin{definition} {\em Most Strings with Few Bad Columns}\\
\noindent Input: $n$ strings $S = \{s_1, \ldots, s_n\}$ of length $\ell$  over an alphabet $\Sigma$, and $k \in \mathbf{Z}_+$.\\
\noindent Solution: a subset $S^* \subseteq S$ of strings such that the number $\{t \in [\ell] \mid \exists s_i^*, s_j^* \in S^*: s_i^*(t) \neq s_j^*(t)\}$ of bad columns is at most $k$.\\
\noindent Objective: maximize $|S^*|$.
\end{definition}

\noindent In other words, a column $t$ is \emph{bad} when its entries are not-all-equal, among strings in $S^*$. The Most Strings with Few Bad Columns Problem generalizes the problem of  finding tandem repeats in a string \cite{LSS}.

\subsection{Our contributions}
A  PTAS for a minimization problem is an algorithm that takes an instance of the problem and a parameter $\epsilon > 0$ and, in time that is polynomial for any fixed $\epsilon$, produces a solution that is within a factor $1 + \epsilon$ of being optimal. An \emph{efficient PTAS} (EPTAS) further restricts the running time to be some function of $\epsilon$ times a constant-degree polynomial in the input size. We present several results on the computational hardness of efficiently finding an approximate solution to the above optimization problems. Specifically, we show the following:
\begin{itemize}
\item The Close to Most Strings Problem has no PTAS, unless $\ZPP=\NP$ (Theorem~\ref{thm:randred}).
\item The Most Strings with Few Bad Columns Problem has no PTAS, unless $\P=\NP$ (Theorem~\ref{t:bcwo2}).
\item We observe that the known PTAS~\cite{ma00} for the Closest to $k$ Strings Problem cannot be improved to an EPTAS, unless $\Wone=\FPT$.
\end{itemize}

Our first result corrects an error in prior literature. A problem is \emph{\APX-hard} if for some fixed $\epsilon>0$, finding a $(1+\epsilon)$-approximation is \NP-hard. A 2000 paper of Ma \cite{ma00} claims that the Close to Most Strings problem is \APX-hard; however, the reduction is erroneous. To explain, it is helpful to define one more problem, {\em Far from Most Strings}, which is the same as Close to Most Strings except that we want to maximize the number of strings $s_i$ in $S$ that satisfy $d(s, s_i) \geq d$ (rather than $\le$). There is considerable experimental interest in heuristics for Far from Most Strings, mostly based on local search~\cite{MOP05,F07,FP11}. Far From Most Strings was introduced and studied by Lanctot et al.~\cite{lanctot}, and they (correctly) showed that for any fixed alphabet size greater than or equal to three, Far from Most Strings is at least as hard to approximate as Independent Set. Currently, Independent Set is known~\cite{KP06} to be inapproximable within a factor of $n/2^{\log^{3/4+\epsilon}n}$ unless $\NP \subset \cls{BPTIME}(2^{\log^{O(1)}n})$.

The main idea in Ma's approach was to consider a binary alphabet. In detail, the Far from Most Strings and Close to Most Strings Problem on alphabets $\Sigma = \{0, 1\}$ are basically the same problem, since a string $s$ of length $\ell$ has distance at most $d$ from $s_i$ if and only if the complementary string $\overline{s}$ has distance at least $\ell-d$ from $s_i$. %For alphabets of size three or more, this equivalence breaks down.
The crucial error in \cite{ma00} is that Ma mis-cited~\cite{lanctot}, assuming that their result worked on binary alphabets. (One reason why the approach of \cite{lanctot} does not extend to binary alphabets in any obvious way is that the instances produced by their reduction satisfy $d = \ell$, whereas Far from Most Strings is easy to solve when $|\Sigma|=2$ and $d = \ell$.)

From \cite{lanctot} and \cite{ma00} we cannot conclude anything about the hardness of Close to Most Strings, nor can we say anything about the hardness of Far from Most Strings when $|\Sigma|=2$. Our results close both of these gaps: the proof of Theorem~\ref{thm:randred} actually shows Close to Most Strings is hard over a binary alphabet, from which it follows that Far from Most Strings is, too. At the same time, the hardness that we are able to achieve is much more modest than the previous claim; we show only that there is no 1.001-approximation. We also require a randomized reduction. It is a very interesting open problem to determine whether this problem has any constant-factor approximation, even over a binary alphabet.

\subsection{Brief Description of Parameterized Complexity}

Some parameterized complexity concepts will arise in later sections, so we give a birds-eye view of this area.  With respect to a parameter $k$, a decision algorithm with running time $f(k)n^{O(1)}$ (where $n$ is the input length) is called {\em fixed parameter tractable} (FPT); the class \FPT\ contains all parameterized problems with FPT algorithms. The corresponding reduction notion between two parameterized problems is an \emph{FPT reduction}, which is FPT, and also increases the parameter by some function that is independent of the instance size. The class \Wone\ is a superset of \FPT\ closed under FPT-reductions. A problem is \emph{\Wone-hard} if any \Wone\ problem can be FPT-reduced to it, and \emph{\Wone-complete} if it is both in \Wone\ and \Wone-hard. There are many natural \Wone-complete problems, like Maximum Clique parameterized by clique size. It is widely hypothesized that $\FPT \subsetneq \Wone$, but unproven, analogous to $\P \subsetneq \NP$.

\section{Approximation Hardness of Close to Most Strings}\label{section:closetomost}

%In this section we prove the following.
\begin{theorem}\label{thm:randred}
For some $\epsilon>0$, if there is a polynomial-time $(1+\epsilon)$-approximation algorithm for the Close to Most Strings Problem, then $\ZPP = \NP$.
\end{theorem}
\begin{proof}
We use a reduction from the {\em Max-2-SAT Problem}, which is to determine for a given 2-CNF formula, an assignment that satisfies the maximum number of clauses. Let $X = \{x_1, \ldots, x_n\}$ be a set of Boolean variables. In 2-CNF, each clause is a disjunction of two literals, each of which is either $x_i$ or $\overline{x_i}$ for some $i$. H{\aa}stad \cite{hastad} showed it is \NP-hard to compute a $22/21$-approximately optimal solution to Max-2-SAT, and this is the starting point for our proof. We will assume that $m \ge n$, i.e.~the number of clauses is greater than or equal to the number of variables, which is without loss of generality since otherwise some variable appears in at most one clause and the instance can be reduced.

\begin{figure}\centering
\includegraphics{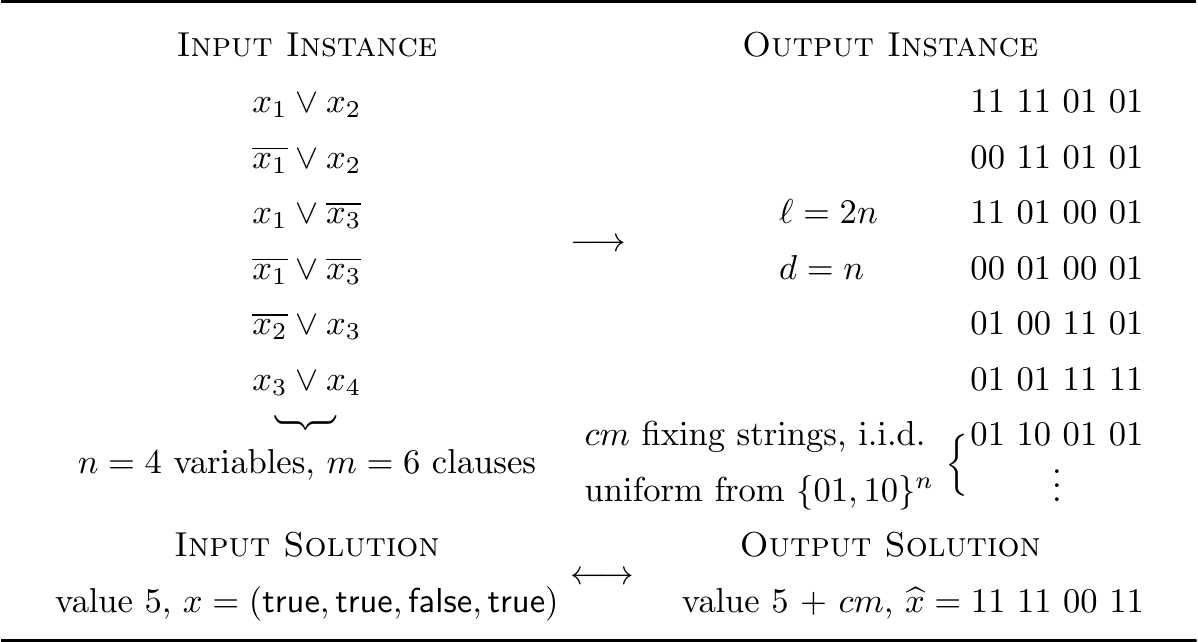}\caption{Overview of the reduction used to prove Theorem~\ref{thm:randred}.}\label{figure:strings}\end{figure}

We give a schematic overview of our reduction in Figure~\ref{figure:strings}. The reduction will be randomized. It takes as input an instance of Max-2-SAT with $m$ clauses and $n$ variables. The reduction's output is an instance of Close to Most Strings with $c m + m$ strings of length $2n$ for some constant $c$, and the distance parameter of the instance is $d=n$. Of these strings, $c m$ will be ``fixing" strings to enforce a certain structure in near-optimal solutions, and the remaining $m$ strings are defined from the clauses as follows. Given a 2-clause $\omega_j$ over the variables in $X$, we define the corresponding string $s_j = s_j(1) \ldots s_j(2n)$ as follows:
\begin{equation*}
s_j(2i - 1) s_j(2i) =
\begin{cases}
00 & \text{if $\omega_j$ contains the literal $\overline{x_i}$, } \\
11 & \text{if $\omega_j$ contains the literal $x_i$, } \\
01 & \text{otherwise. } \\
\end{cases}
\end{equation*}
\noindent The fixing strings will all be elements of $\{01, 10\}^n$, selected independently and uniformly at random.

We now give a high-level explanation of the proof. For every variable assignment vector $x$ define a string $\widehat{x}$ via
\begin{equation*}
\widehat{x}(2i - 1) \widehat{x}(2i) =
\begin{cases}
11 & \text{if $x_i$ is true, } \\
00 & \text{if $x_i$ is false. } \\
\end{cases}
\end{equation*}
Notice that $\widehat{x}$ is at distance exactly $d=n$ from all of the fixing strings, and that $d(\widehat{x}, s_j) \le n$ if and only if $x$ satisfies clause $\omega_j$. Hence, if $x$ satisfies $k$ clauses, the string $\widehat{x}$ is within distance $d$ of $c m + k$ out of the $c m + m$ total strings. We will show conversely that with high probability, for all strings $s$ within distance $d$ of $c m$ of the strings, we have $s \in \{00, 11\}^n$. Using this crucial structural claim, it follows that any sufficiently good approximation algorithm for Close to Most Strings must output $s$ such that $s = \widehat{x}$ for some $x$. Then the claim will be complete via standard calculations.

Here is the precise statement of the structural property.
\begin{lemma}\label{lemma:fixing}
For $c \ge 20$, the following holds. Let $F$ be a set of $c m$ strings selected uniformly and independently at random from $\{01, 10\}^n$ (with replacement), with $m \ge n$. Then with probability at least $1-0.9^{n}$, every string $s \in \{0, 1\}^{2n} \setminus \{00, 11\}^n$ has distance greater than $n$ from at least $m$ strings in $F$.
\end{lemma}
\begin{proof}
To explain the proof more simply, fix $s$ and consider a particular $f \in F$. By hypothesis, this $s$ satisfies $s(2i-1) \neq s(2i)$, say $s(2i-1)=0$ and $s(2i)=1$ (the other case is symmetric). Since $f$ is chosen uniformly at random from $\{01, 10\}^{n}$, the event $\mathcal{E}$ where $f(2i-1)=1$ and $f(2i)=0$ has $\Pr[\mathcal{E}]=1/2$. A short calculation which we postpone momentarily shows that $\Pr[d(s, f) \ge n+1 \mid \mathcal{E}] \ge 1/2$. So unconditioning, $\Pr[d(s, f) \ge n+1] = \Pr[d(s, f) \ge n+1 \mid \mathcal{E}] \cdot \Pr[\mathcal{E}] \ge 1/4$.

Let us verify now that $\Pr[d(s, f) \ge n+1 \mid \mathcal{E}] \ge 1/2$. Observe that $d(s, f)$ is a sum of $n$ independent random variables $d(s(2j-1)s(2j), f(2j-1)f(2j))$ for $j$ from 1 to $n$; conditioning on $\mathcal{E}$ just fixes one of these variables at 2. The remaining ones are either always 1 (if $s(2j-1)=s(2j)$), or a uniformly random element of $\{0, 2\}$. The conditioned random variable $d(s, f) \mid \mathcal{E}$ is thus a shifted and scaled binomial distribution, in particular it is symmetric about $n+1$. So
$\Pr[d(s, f) \ge n+1 \mid \mathcal{E}] = \Pr[d(s, f) \le n+1 \mid \mathcal{E}]$ and since these two probabilities' sum is at least 1,
$\Pr[d(s, f) \ge n+1 \mid \mathcal{E}] \ge 1/2$ follows.

We will continue reasoning about this fixed $s$, and use a Chernoff bound to get large enough probability to work for all possible $s$. Let $F = \{f_1, \dotsc, f_{c m}\}$ and let $X_i$ be an indicator variable for the event that $d(f_i, s) > n$. We have argued that each $X_i$ is 1 with probability at least 1/4. Therefore, $E[\sum_i X_i] \ge cm/4$. We will use a Chernoff bound of the following form:
\begin{claim}[Lower Chernoff bound, \cite{randalgs}]
If $X$ is a sum of independent 0-1 random variables, then we have $$\Pr[X < (1-\delta)E[X]] < \exp(-E[X] \delta^2 / 2).$$
\end{claim}
Choose $\delta$ so that $(1-\delta)cm/4 = m$, i.e.~$\delta=1-4/c$. Then $\Pr[X < m] \le \Pr[X < (1-\delta)E[X]]  < \exp(-cm/4 \cdot (1-4/c)^2/2) = \exp(\frac{-(c-4)^2}{8c}m)$. By a union bound over all $4^n-2^n$ possible choices of $s$, the probability that a random choice of $F$ admits \emph{any} bad $s$ is at most $$(4^n-2^n) \exp\Bigl(\frac{-(c-4)^2}{8c}m\Bigr) < 4^n \exp\Bigl(\frac{-(c-4)^2}{8c}n\Bigr) = \exp\Bigl(\bigl(\ln 4 - \frac{(c-4)^2}{8c}\bigr)n\Bigr).$$
Any large enough $c$ makes this exponentially decreasing in $n$; it is straightforward to calculate that when $c=20$
this is at most $0.9^n$, as needed.
\end{proof}

Now let us complete the overall proof; fix $c=20$. Given a Max-2-SAT instance, we run the randomized reduction above to get an instance of Close to Most Strings. Let $s_A$ be a $(1+\epsilon)$-approximation for this instance, where $\epsilon$ will be a small constant fixed later to satisfy two properties.

Let $k^*$ be the maximum number of satisfiable clauses in the Max-2-SAT instance. As an important technicality, note that $k^*$ is lower-bounded by $m/2$, since the expected number of clauses satisfied by a random assignment is at least $m/2$, by linearity of expectation. So the optimal solution to the Close to Most Strings instance has value at least $cm+m/2$.

First we want to use the structural lemma (Lemma \ref{lemma:fixing}). Assume for now the bad event with probability $0.9^n$ does not happen; so every $s \not\in \{00, 11\}^n$ (i.e.~not of the form $s = \widehat{x}$) is within distance $d$ of at most $cm$ of the $(c+1)m$ strings. Thus provided that $\epsilon$ is small enough to satisfy $1+\epsilon < \frac{cm+m/2}{cm} = 1 + \frac{1}{2c}$, then $s_A$ is of the form $\widehat{x}_A$ for some $x_A$.

Next we finish the typical calculations in a proof of \APX-hardness. We know that $s_A$ is within distance $d$ of at least $(cm+k^*)/(1+\epsilon)$ strings. If we can pick $\epsilon$ so that
\begin{equation}\label{eq:hoopla}\frac{cm+k^*}{1+\epsilon} > cm + \frac{21}{22}k^*\end{equation}
then $\widehat{x}_A$ satisfies more than $\frac{21}{22}k^*$ clauses, which is NP-hard by H{\aa}stad's result. Using that $k^* \ge m/2$, it is easy to verify that \eqref{eq:hoopla} holds for all $\epsilon < 1/(21+44c)$.

Finally, we confirm that the randomized algorithm for Max-2-SAT coming from the reduction is ZPP-style, i.e.~Las Vegas style. As long as the output $s_A$ of the Close to Most Strings approximation algorithm satisfies $s_A \not\in \{00, 11\}^n$ we re-create the reduction again using fresh random bits and re-run the approximation algorithm. But once $s_A \in \{00, 11\}^n$ we know for certain that $\widehat{x}_A$ is a 22/21-approximate solution for Max-2-SAT, as needed. The expected number of trials is at most $1/(1-0.9^n) = O(1)$.
\end{proof}

\section{Non-existence of an EPTAS for Closest to $k$ Strings}\label{section:ctoks}

Ma showed in \cite{ma00} that the Closest to $k$ Strings problem has a PTAS, which contrasts with the \APX-hardness we obtain for the other problems in this paper. A natural question that comes up after a PTAS is obtained, is whether the running time can be improved to an EPTAS, or even further to a FPTAS (running time polynomial in the input length and $\epsilon^{-1}$). \defer{~\cite{Arora98,HuntMRRRS98,marx05}.}
We observe there does not exist an EPTAS for Closest to $k$ Strings when the alphabet is unbounded, unless $\Wone=\FPT$. To see this, we use a well-known fact relating fixed-parameter algorithms to the notion of an EPTAS, e.g.~see~\cite{marx_survey}, along with the fact that the decision version of Closest to $k$ Strings is \Wone-hard when parameterized by $d$~\cite{boucher_ma}.

In detail, suppose for the sake of contradiction that we had an EPTAS for Closest to $k$ Strings, i.e.~that one could obtain a $(1+\epsilon)$-approximation in time $f(\epsilon)s^{O(1)}$ where $s$ is the input size. It is enough to prove that there is an FPT algorithm for the decision version of Closest to $k$ Strings, with parameter $d$. Given an instance of this parameterized problem we need only call the EPTAS with any $\epsilon$ less than $(d+1)/d$; notice the resulting algorithm takes FPT time with respect to $d$. To analyze this, let $d_{ALG}$ be the distance value of the solution produced by the EPTAS algorithm, and $d_{OPT}$ be the optimal distance value. If $d_{OPT} \le d$, since $d_{OPT} \le d_{ALG} \le (1+\epsilon)d_{OPT}$ and $d_{OPT}, d_{ALG} \in \mathbf{Z}$, we have $d_{OPT} = d_{ALG} \le d$. Otherwise, $d_{ALG} \ge d_{OPT} > d$. So, we get an FPT algorithm just by comparing $d_{ALG}$ to $d$.

%Suppose $\epsilon = \frac{1}{2\alpha}$, where $\alpha$ is the value of the objective function, then a $(1+\epsilon)$-approximation %algorithm would distinguish between ``yes'' and ``no'' instances of the problem, and an EPTAS would be a $O(f(\epsilon)n^{O(1)}) = %O(g(\alpha)n^{O(1)})$-time algorithm for the problem. If a problem does not admit a FPT algorithm parameterized by the value of the %objective function (unless FPT=W[1]) then the corresponding optimization problem does not admit an EPTAS (unless FPT=W[1]).  Hence, the %following observation follows from This implies that the PTAS given by Ma \cite{ma00} is the best we can hope for.

\begin{observation} Closest to $k$ Strings has no EPTAS unless $\Wone = \FPT$. \end{observation}

\def\vz{\mathbf{0}}

\section{\APX-Hardness of Most Strings with Few Bad Columns}

In this section, we prove that the Most Strings with Few Bad Columns Problem is APX-hard, even in binary alphabets. To do this we reduce from the {\em Densest-$k$-Subgraph Problem}: given a graph $G = (V, E)$ and a parameter $k$, find a subset $U \subseteq V$ with  $|U|=k$ such that $|E[U]|$ is maximized --- here $E[U]$ denotes the \emph{induced edges} for $U$, meaning the set of all edges with both endpoints in $U$.

Our reduction will be approximation-preserving up to an additive $+1$ term. Given an instance $(G=(V, E), k)$ of Densest-$k$-Subgraph, we will generate an instance of Most Strings with Few Bad Columns with $|E|+1$ strings, each of length $|V|$, and with the same values for the two parameters $k$ (size of subgraph, maximum number of bad columns). %We assume that $k > 2$ since $k \leq 2$ produces trivial cases.

%\begin{SCfigure}
 % \centering
  %\includegraphics[width=0.5\textwidth]{example}
%\caption{Example of the reduction from an instance of Densest-$k$-Subgraph with $G$ and $k = 3$ to an instance of Most Strings with Few Bad Columns with 6 strings of length 5.}
%\label{fig:example}
%\end{SCfigure}

\begin{figure}
\centering
  \includegraphics{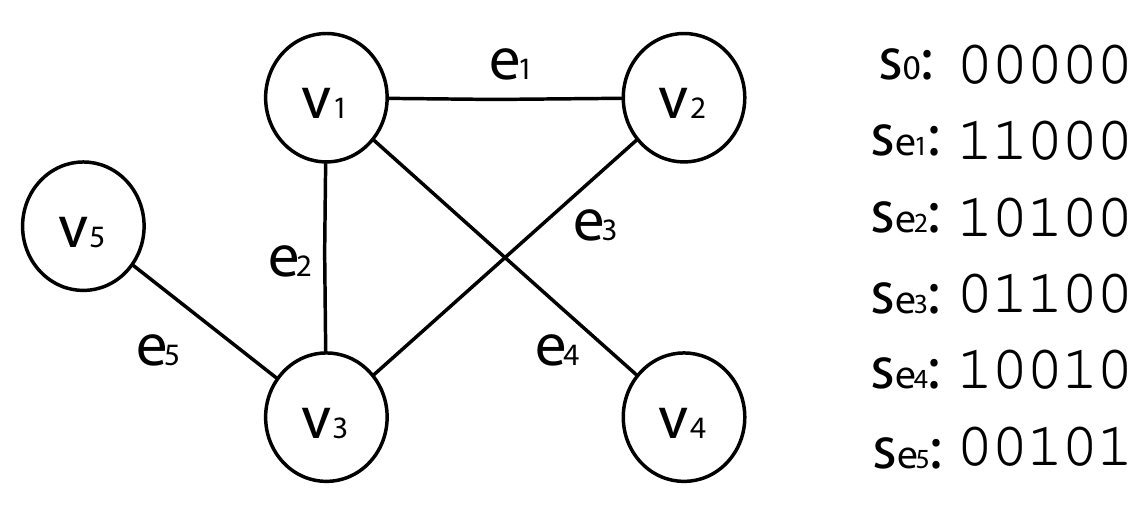}
\caption{Example of the reduction from an instance of Densest-$k$-Subgraph with $G$ and $k = 3$ to an instance of Most Strings with Few Bad Columns with 6 strings of length 5.}
\label{fig:example}
\end{figure}

Let us define the set $S$ of strings generated by the reduction; to do this, index $V = \{v_1, v_2, \dotsc\}.$ For each edge $e = v_iv_j \in E$, let that edge's \emph{0-1 incidence vector} $\chi(e)$ be the 0-1 string with 1s in positions $i$ and $j$ and 0 elsewhere; we put $\chi(e)$ into $S$. Finally, we put one more string into $S$, namely the all-zero string $\vz$. This completes the description of the reduction; note it only takes polynomial time. See Figure~\ref{fig:example} for an illustration of this reduction.

\begin{claim}\label{claim:optval}
Let $\alpha$ be the optimal value for the Densest-$k$-Subgraph instance. Then the optimal value $\beta$ for the new Most Strings with Few Bad Columns instance is $\beta = \alpha+1$.
\end{claim}
\begin{proof}

%We describe how to generate a set $S$ of $|E| + 1$ strings, each of length $|V|$, such that $G$ has a set $U$ of $k$ vertices with the maximum number (say $\alpha$) of induced edges if and only if there is a subset of $S$ of size $\alpha + 1$, denoted as $\overline{S}$, which has at most $k$ bad columns.

%We restrict our interest to the Most Strings with $k$ Identical Columns Problem on binary alphabet. The reduction generates a set of size $|E|$ of length-$|V|$ strings $S_{|E|} = \{s_1, \ldots, s_{|E|}\}$, which encodes the edges in $G$. A string $s_i$ in $S_{|E|}$ corresponds to an edge $e_i = (v_j, v_m)$ in $G$ and is defined to be 1 at positions $j$ and $m$, and 0 at the remaining $|V| - 2$ positions. We let $S = S_{|E|} \cup {\vz}$.

First we show the easy direction, that $\beta \ge \alpha+1$. Consider the optimal $U$ for Densest-$k$-Subgraph, so that $|E[U]|=\alpha$ and $|U|=k$. Define a subset $\overline{T}$ of $S$ by $\overline{T} = \{\vz\} \cup \{\chi(e) \mid e \in F\}$. Then the strings in $\overline{T}$ are all zero on any index corresponding to a node outside of $V$; the only bad columns are those corresponding to nodes in $V,$ of which there are only $k$. So $\beta \ge |\overline{T}|=\alpha+1$.

For the reverse direction, take a subset $T$ of $\beta$ strings that have at most $k$ bad columns. We can assume without loss of generality that the string $\vz$ is in $T$, as the following structural lemma shows.

\begin{lemma}\label{l:zeros}
Let $T \subseteq S$ be a subset of strings with at most $k$ bad columns. Then there is a subset $T'$ of $S$ with at most $k$ bad columns, $|T'| \ge |T|$, and $\vz\in T'$.
\end{lemma}

Assume for the moment that the lemma is true. Then we simply reverse the above reduction to show $\alpha \ge \beta-1$. Take an optimal set $S^*$ of strings with $|S^*|=\beta$ and such that $S^*$ has at most $k$ bad columns. By Lemma~\ref{l:zeros} we may assume $\vz \in S^*$ --- this implies that the set $J$ of all non-bad columns for $S^*$ satisfies $s(j) = 0$ for all $s \in S^*, j \in J$. Thus, each $\chi(uv) \in S^* \setminus \{\vz\}$ has both of its 1s appearing at positions in $[\ell] \setminus J$, or equivalently each such $uv$ is an element of $E[V \setminus J]$. So $V \setminus J$ is the required solution for Densest-$k$-Subgraph, and it has at least $\beta-1$ induced edges.

\begin{proof}[Proof of Lemma~\ref{l:zeros}]
Assume that $\vz \not \in T$, otherwise the lemma trivially follows. Also, assume $W = T\cup \{\vz\}$ has more than $k$ bad columns, otherwise we can take $T' = W$. Thus there must be a column that is not bad for $T$ but that becomes bad when adding $\vz$. I.e.~$T$ has a column that is entirely 1s. It follows that, viewed in the original graph setting, there exists a vertex $v$ that is an end-point of all the edges corresponding to $T$. Pick any such edge arbitrarily, i.e.~suppose $s = \chi(vw) \in T$. Since the input graph is simple, in column $w$, all entries of $T$ are 0 except for $\chi(vw)$. Hence, $T' = T \setminus s \cup \vz$ satisfies the lemma: compared with $T$ it is bad in column $v$ but not bad in column $w$.
\end{proof}

%Note that in our construction edges are different and thus every row contributes a different pair of 1's indices in the strings of $S^*$. Therefore, if $|E|>|V|-1$, there cannot be a vertex that is an end-point of all edges in the graph. We may, therefore assume that $|E|\leq |V|-1$. In this case, every bad column in $S^*$ has exactly one 1. So removing any $s_{e_{\ell i}}$ from $S^*\cup \{\vz\}$ gives a subset $S'$ with $\alpha+1$ strings and $k$ bad columns.
%\qed
This ends the proof of Claim~\ref{claim:optval}.
\end{proof}

%\begin{figure}[h]
%\begin{center}
%\includegraphics[width=70mm]{example}
%\caption{Example of the reduction from an instance of Densest-$k$-Subgraph with $G$ and $k = 3$ to an instance of Most Strings with Few Bad Columns with 6 strings of length 5.}
%\label{fig:example}
%\end{center}
%\end{figure}

%As a consequence of this reduction we get the following theorem.
%\end{proof}
This reduction yields our result:

\begin{theorem}\label{t:bcwo2}
The Most Strings with Few Bad Columns Problem is NP-hard, and APX-hard.
\end{theorem}
\begin{proof}
Khot~\cite{khot} showed that the Densest-$k$-Subgraph Problem is APX-hard. We need only to argue that our reduction can transform a PTAS for Most Strings with Few Bad Columns into a PTAS for Densest-$k$-Subgraph. Indeed, if we had a $(1+\delta)$-approximation algorithm for Most Strings with Few Bad Columns, then we get an algorithm for Densest-$k$-Subgraph that always returns a solution of value at least
$$(OPT+1)/(1+\delta)-1 = (OPT-\delta)/(1+\delta) \ge OPT(1-\delta)/(1+\delta) = OPT/(1+O(\delta))$$
where we used $OPT \ge 1$ in the middle inequality.
\end{proof}

While we ruled out a PTAS, it would also be out of the reach of current technology to obtain a constant or polylogarithmic factor for Most Strings with Few Bad Columns, because the best known approximation factor for the Densest-$k$-Subgraph Problem is $O(|V|^{1/4 + \epsilon})$~\cite{BCCFV}.

%\begin{theorem}\label{t:bcwo2}
%The Most Strings with $k$ Identical Columns Problem has no PTAS unless P=NP.
%\end{theorem}

\section{Conclusions and Open Problems}

Our results demonstrate that while outliers help model the problems associated with using biological data, such problems are computationally intractable to approximate. 
Here are the main open problems related to our results:
\begin{itemize}
\item Is there a constant-factor approximation for either Close to Most Strings or  Most Strings with Few Bad Columns (even over a binary alphabet)?
%\item Is there a constant-factor approximation for Most Strings with Few Bad Columns ?
\item Does there exist an EPTAS for the Closest String Problem?  Since the Closest String Problem is FPT with respect to $d$~\cite{GNR03}, the standard technique used in Section~\ref{section:ctoks} cannot be used naively.
\item Does there exist an EPTAS for the Closest to $k$ Strings Problem over a bounded-size or binary alphabet? The reduction used in Section~\ref{section:ctoks} needs an arbitrarily large alphabet.
%\item The approximability of other string selection problems with outliers warrants investigation. The {\em Sum Closest String with Outliers} problem takes a set $S$ of same-length strings as input, and seeks a maximum-size $S^* \subseteq S$ subject to a given upper bound on $\sum_{s_i \in S^*} d(s, s_i).$ This problem is closely related to the problems of finding cycle detection, correction and approximate periodicity~\cite{ami,levy}. What is the approximability of this problem?
\end{itemize}

%Does either Close to Most Strings or Most Strings with Few Bad Columns have a constant-factor approximation algorithm? This is open even for binary alphabets.

%Does Closest to $k$ Strings have an EPTAS when the alphabet is binary? The reduction used in Section~\ref{section:ctoks} needs an arbitrarily large alphabet. A more important problem is, does there exist an EPTAS for the Closest String Problem? \defer{Since  Li et al.~\cite{LMW02} first proved the existence of a PTAS, there has been significant effort given to improving on the running time of the PTAS.  As previously mentioned,} Since the Closest String Problem is FPT with respect to $d$~\cite{GNR03}, the standard technique used in Section~\ref{section:ctoks} cannot be used naively.

%The {\em Sum Closest String with Outliers} problem takes a set $S$ of same-length strings as input, and seeks a maximum-size $S^* \subseteq S$ subject to a given upper bound on $\sum_{s_i \in S^*} d(s, s_i).$ This problem is closely related to the problems of finding cycle detection, correction and approximate periodicity~\cite{ami,levy}. What is the approximability of this problem?

\subsection*{Acknowledgments}

The authors would like to thank Dr. Bin Ma for mentioning the error in his inapproximability proof and encouraging us to work on a correction, and Dr. Daniel Lokshtanov and Christine Lo for their insights and comments. %CB is supported by NSERC PDF and the Gerald Schwartz and Heather Reisman Foundation. AL is partly supported by ISF grant 347/09.

%%%%%%%%%%%%%%%%%%%%%%%%%%%%%%%%%%%%%%%%%%%%%%%%%%%%%%%%%%%%%%%%%%%%%%%%%%%%%
\bibliographystyle{plain}
%\bibliography{08_12_2011}

%\newpage
%\input{appendix.tex}
\vspace{-5mm}

\end{document}